\documentclass{article}

\usepackage{amsfonts,amssymb,amsmath,a4wide,paralist}
\usepackage{epsfig,rotating,framed,enumerate,listings}

\newcommand{\OPEN}{{\it open}}

\newcommand{\CUT}{{\it cut}}
\newcommand{\PL}{{\it plt}}
\newcommand{\PLS}{{\it plts}}
\newcommand{\FI}{{\it first}}
\newcommand{\LA}{{\it last}}
\newcommand{\links}{{\it left}}
\newcommand{\rechts}{{\it right}}

\newcommand{\FRONT}{{\it front}}
\newcommand{\NN}{{\mathbb N}}
\newcommand{\ZZ}{{\mathbb Z}}
\newcommand{\bigo}{\ensuremath{\mathcal{O}}}
\newtheorem{theorem}{Theorem}[section]
\newtheorem{example}[theorem]{Example}
\newtheorem{lemma}[theorem]{Lemma}
\newtheorem{proposition}[theorem]{Proposition}
\newtheorem{corollary}[theorem]{Corollary}
\newcommand{\dpw} {\text{d-pw}}
\newenvironment{proof}{\noindent{\bf Proof~}}{\null\hfill $\Box$\par\medskip}

\newtheorem{remark}[theorem]{Remark}
\newenvironment{desctight}
  {\begin{list}{}{
\setlength\labelwidth{-5pt}
        \setlength{\itemsep}{0.5pt}
        \setlength{\parsep}{0pt}
        \setlength\itemindent{-\leftmargin}
        
}}
    {\end{list}}

\begin{document}

\title{Complexity of the FIFO Stack-Up Problem\thanks{A short version of 
this paper appeared in the Proceedings of the {\em International Conference on 
Operations Research} (OR 2013) \cite{GRW14a}.}}

\author{Frank Gurski\thanks{University of  D\"usseldorf,
Institute of Computer Science, Algorithmics for Hard Problems Group, 40225 D\"usseldorf, Germany,
{\tt frank.gurski@hhu.de}}
\and 
Jochen Rethmann\thanks{Niederrhein University of Applied Sciences,
Faculty of Electrical Engineering and Computer Science, 47805 Krefeld,
Germany, 
{\tt jochen.rethmann@hs-niederrhein.de}}
\and
Egon Wanke\thanks{University of  D\"usseldorf,
Institute of Computer Science,  Algorithms and Data Structures Group, 40225 D\"usseldorf, Germany,
{\tt e.wanke@hhu.de}
}}

\maketitle

\begin{abstract}
We study the combinatorial FIFO stack-up problem. In delivery industry,
bins have to be stacked-up from conveyor belts onto pallets with respect to
customer orders. Given $k$
sequences $q_1, \ldots, q_k$ of labeled bins and a positive integer $p$,
the aim is to stack-up the bins by iteratively removing the first bin
of one of the $k$ sequences and put it onto an initially empty pallet of
unbounded capacity
located at one of $p$ stack-up places. Bins with different pallet labels 
have to be placed on different pallets, bins with the same pallet label 
have to be placed on the same pallet. After all bins for a pallet have 
been removed from the given sequences, the corresponding stack-up place 
will be cleared and becomes available for a further pallet. The 
FIFO stack-up problem is to find a stack-up sequence such that all pallets 
can be build-up with the available $p$ stack-up places. 

In this paper, 
we introduce two digraph models for the FIFO stack-up problem, namely
the processing graph and the sequence graph. 
We show that there is a processing of some list of sequences with at most 
$p$ stack-up places if and only if the sequence graph of this list has 
directed pathwidth at most  $p-1$. This connection implies that  
the FIFO stack-up problem is NP-complete
in general, even if there are at most 6 bins for every pallet and
that the problem can be solved in polynomial time, if the number
$p$ of stack-up places is assumed to be fixed. 
Further the processing graph allows us to show that
the problem can be solved in polynomial time, if the number $k$ of 
sequences is assumed to be fixed.

\bigskip
\noindent
{\bf Keywords:} 
computational complexity;  combinatorial optimization;
directed pathwidh; discrete algorithms
\end{abstract}

\section{Introduction}

We consider the combinatorial problem of stacking up bins from
conveyor belts onto pallets. This problem originally appears in
{\em stack-up systems} that play an important role in delivery industry
and warehouses. Stack-up systems are often the back end of
{\em order-picking systems}. A detailed description of the applied background of
such systems is given in \cite{Kos94,RW97}. Logistic experiences over 30
years lead to high flexible conveyor-based stack-up systems in delivery
industry \cite{Yam09}. We do not intend to modify the architecture of
existing systems, but try to develop efficient algorithms to control them. 

The bins that have to be stacked-up onto pallets reach the
stack-up system on a main conveyor belt. At the end of the line
they enter the palletizing system. Here
the bins are picked-up by stacker cranes or robotic arms and moved
onto pallets, which are located at {\em stack-up places}. 
Often vacuum grippers are used to pick-up the bins. This picking process can
be performed in different ways depending on the architecture
of the palletizing system (single-line or multi-line palletizers).
 Full pallets are carried away by automated guided vehicles, or
by another conveyor system, while new empty pallets are placed
at free stack-up places. 

The developers and producers of robotic palletizers distinguish between single-line
and multi-line palletizing systems. Each of these systems has its advantages and disadvantages.

In {\em single-line palletizing systems} there is only one conveyor belt from
which the bins are picked-up. Several robotic arms or stacker cranes are
placed around the end of the conveyor. We model such systems by a random
access storage which is automatically replenished with bins from the main conveyor, 
see Figure \ref{F10a}.
The area from which the bins can be picked-up is called the storage area. It
is determined by the coverage of stacker cranes or robotic arms.

\begin{figure}[hbtp]
\centerline{\epsfxsize=90mm \epsfbox{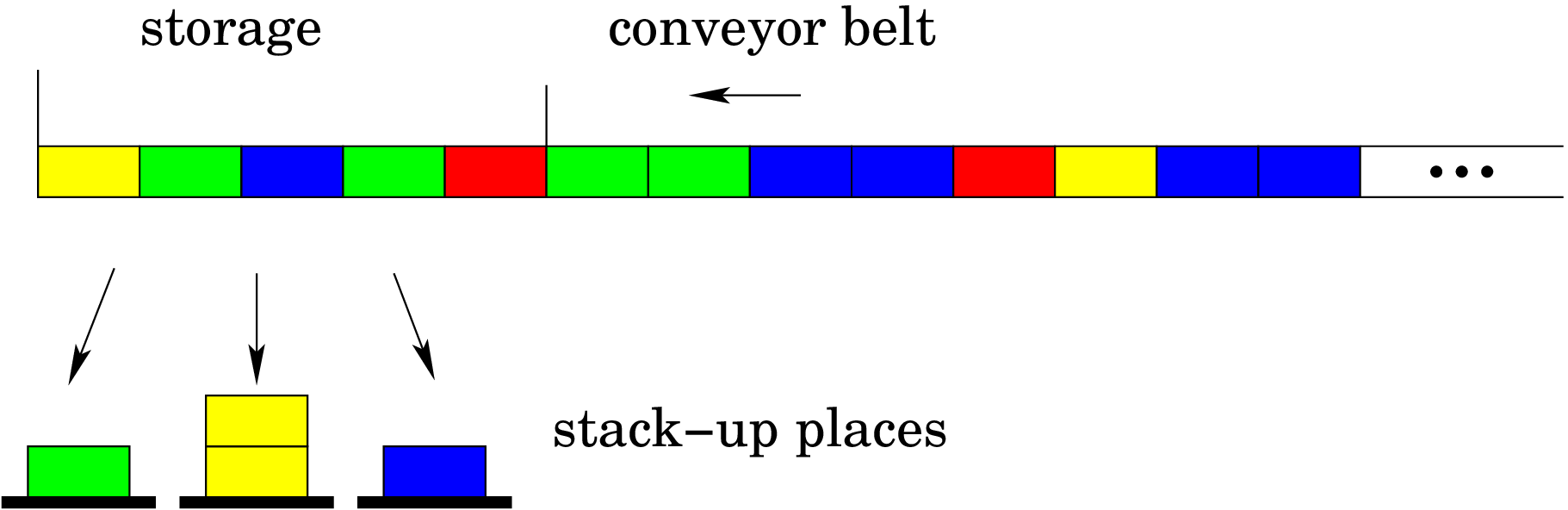}}
\caption{The single-line stack-up system using a random access
storage of size 5. The colors represent the pallet labels. 
Bins with different colors have to be placed on different pallets, bins 
with the same color have to be placed on the same pallet.}
\label{F10a}
\end{figure}

In {\em multi-line palletizing systems} there are several buffer conveyors from which
the bins are picked-up. The robotic arms or stacker cranes are placed at
the end of these conveyors. Here, the bins from the main conveyor of the order-picking
system first have to be distributed to the multiple infeed lines to enable
parallel processing. Such a distribution can be done by some cyclic storage conveyor,
see Figure \ref{F00}. From the cyclic storage conveyor the bins are pushed out to the buffer conveyors.
A stack-up system using a cyclic storage conveyor is, for example, located at
Bertelsmann Distribution GmbH in G\"utersloh, Germany. On certain days,
several thousands of bins are stacked-up using a cyclic storage conveyor with
a capacity of approximately 60 bins and 24 stack-up places, while up to
32 bins are destined for a pallet.
This palletizing system has originally initiated our research.

\begin{figure}[hbtp]
\centerline{\epsfxsize=65mm \epsfbox{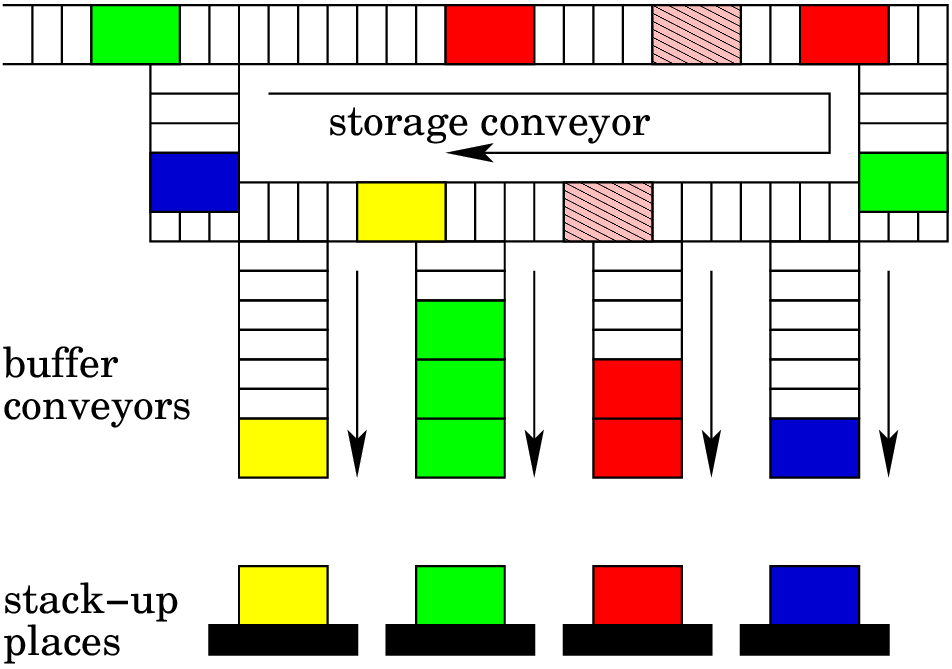}}
\caption{A multi-line stack-up system with a pre-placed cyclic storage conveyor.}
\label{F00}
\end{figure}

If we ignore the task to distribute the bins from the main
conveyor to the $k$ buffer conveyors, i.e., if the filled buffer conveyors are already given,
and if each arm can only pick-up the first bin of one of the buffer conveyors,
then the system is called a FIFO palletizing system.
Such systems can be modeled by several simple queues, see Figure \ref{F01}.

\begin{figure}[hbtp]
\centerline{\epsfxsize=65mm \epsfbox{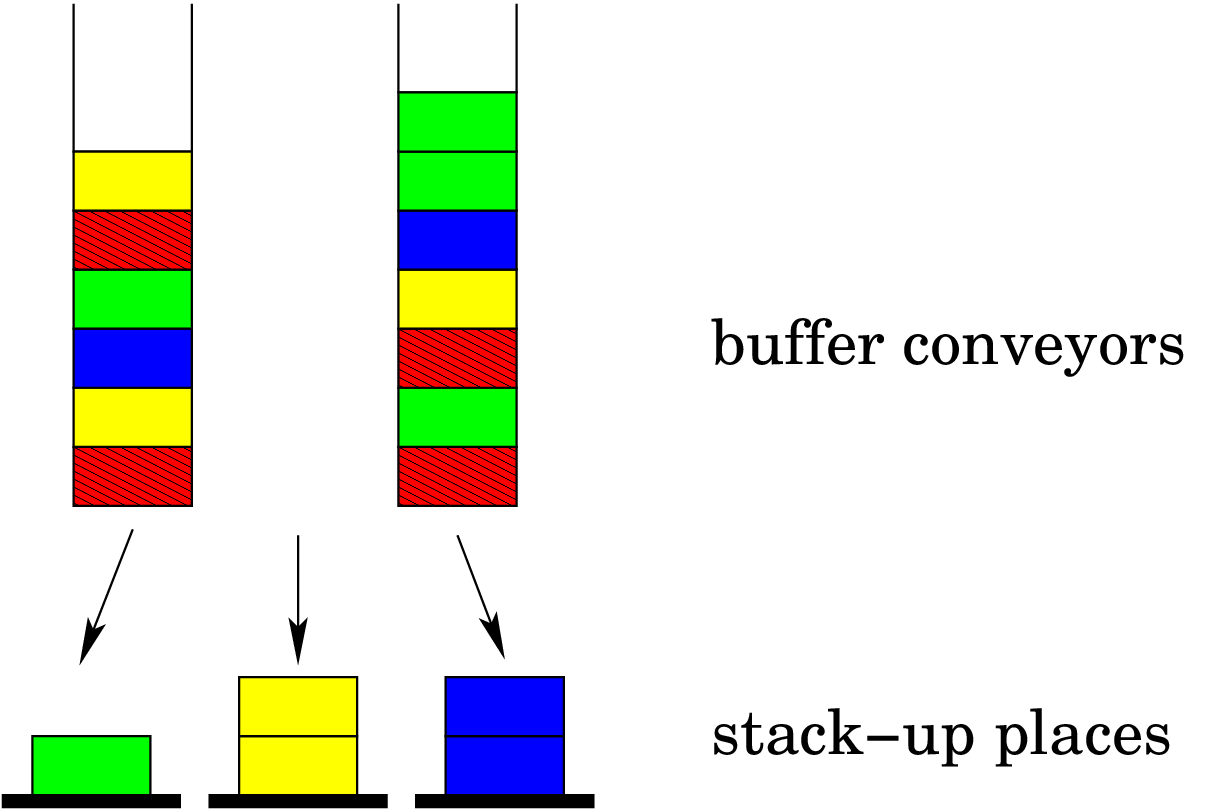}}
\caption{The FIFO stack-up system analyzed in this paper. 
The system is blocked (cf.~page~\pageref{page-blocking} for the definition), 
because the pallet for the red bins cannot be opened and the pallets for the 
green, yellow, and blue bins cannot be closed.}
\label{F01}
\end{figure}

From a theoretical point of view, an instance of the FIFO stack-up problem
consists of $k$ sequences $q_1, \ldots, q_k$ of bins and a number of available
stack-up places $p$. Each bin of each sequence $q_i$ is destined for exactly one pallet.
The FIFO stack-up problem is to decide whether one can remove iteratively the bins
from the sequences such that in each step only one of the first bins of $q_1, \ldots, q_k$
is removed and after each step at most $p$ pallets are open.
A pallet $t$ is called {\em open}, if at least
one bin for pallet $t$ has already been removed from the sequences,
and if at least one bin for pallet $t$ is still contained
in the remaining sequences. If a bin $b$ is removed from a sequence
then all bins located behind $b$ are moved-up one position to the front
(cf.~Section \ref{sec-prel} for the formal definition).

Every processing should absolutely avoid blocking situations.\label{page-blocking}
A system is {\em blocked}, if all stack-up places are occupied by
pallets, and non of the bins that may be used in the next step are destined for an open
pallet. To unblock the system,
bins have to be picked-up manually and moved to pallets by human workers.
Such a blocking situation is sketched in Figure \ref{F01}.

The single-line stack-up problem can be defined in the same way. An
instance of the single-line stack-up problem consists of one sequence
$q$ of bins, a storage capacity $s$, and a number of available stack-up places $p$.
In each step one of the first $s$ bins of $q$ can be removed. Everything else is defined analogously.

Many facts are known about single-line stack-up systems \cite{RW97,RW00,RW01}.
In \cite{RW97} it is shown that the single-line stack-up decision problem
is NP-complete, but can be solved efficiently if the storage capacity $s$
or the number of available stack-up places $p$ is fixed. The problem remains NP-complete as shown in
\cite{RW00}, even if the sequence contains at most 9 bins per pallet.
In \cite{RW00}, a polynomial-time {\em off-line} approximation algorithm
for minimizing the storage capacity $s$ is introduced.
This algorithm yields a solution that is optimal up to a factor bounded by $\log(p)$.
In \cite{RW01} the performances of simple {\em on-line} stack-up
algorithms are compared with optimal off-line solutions by a competitive
analysis \cite{Bor98,FW98}. 


\medskip
The FIFO stack-up problem has algorithmically not been studied by
other authors up to now, although stack-up systems play an important
role in delivery industry and warehouses \cite{Yam09}.
In our studies, we neither limit the number of bins for a pallet nor restrict
the number of stack-up places to the number of buffer conveyors.
That is, the number of stack-up
places can be larger than or less than the number of buffer conveyors.

In Section 3, we introduce two digraph models that help us to find algorithmic
solutions for the FIFO stack-up problem. The first digraph is the {\em processing graph}.
It has a vertex for every possible configuration of the system and an arc from configuration $A$
to configuration $B$ if configuration $B$ can be obtained from configuration $A$
by a single processing step. If the number $k$ of sequences is assumed to be fixed,
a specific search strategy on the processing graph
allows us to find an optimal solution for the FIFO stack-up problem
in polynomial time.
For this case, we additionally give a non-deterministic
algorithm that uses only logarithmic work-space. 

The second digraph is called the {\em sequence graph}.
It has a vertex for every pallet and an arc from pallet $a$ to pallet $b$,
if and only if in any processing pallet $b$ can only be closed if pallet $a$
has already been opened.
We show in Section 4 how the sequence graph can be constructed and
that there is a processing of at most 
$p$ stack-up places if and only if the sequence graph has 
directed pathwidth at most  $p-1$ (cf. Section  \ref{seq-gr}
for the formal definition of directed pathwidth). This implies that 
the FIFO stack-up problem can be solved
in polynomial time, if the number $p$ of given stack-up places is 
assumed to be fixed. In Section 5, this relationship is used to show
that the FIFO stack-up problem is NP-complete in general, even
if all sequences together contain at most 6 bins destined for the same pallet.

\section{Preliminaries}\label{sec-prel}

Unless otherwise stated, $k$ and $p$ are some positive integers throughout
the paper. We consider {\em sequences} $q_1=(b_1, \ldots, b_{n_1}),
\ldots, q_k=(b_{n_{k-1}+1}, \ldots, b_{n_k})$ of pairwise distinct {\em
bins}. These sequences represent the buffer queues (handled by the
buffer conveyors) in real stack-up systems.
Each bin $b$ is labeled with a {\em pallet symbol} $\PL(b)$ which
is, without loss of generality, some positive integer. We say bin $b$ is
destined for pallet $\PL(b)$.
The labels of the pallets can be chosen arbitrarily, because we only need
to know whether two bins are destined for the same pallet or for different
pallets. In order to avoid confusion between indices (or positions) and pallet symbols,
we use in our examples characters for pallet symbols. The set
of all pallets of the bins in some sequence $q_i$ is denoted by
\[
\PLS(q_i)= \{ \PL(b) ~|~ b \in q_i \}.
\]
For a list of sequences $Q = (q_1, \ldots, q_k)$ we denote
$$\PLS(Q) = \PLS(q_1) \cup \cdots \cup \PLS(q_k).$$

For some sequence $q = (b_1, \ldots, b_n)$, we say bin $b_i$ is {\it on the
left of} bin $b_j$ in sequence $q$ if $i < j$.
A sequence $q' = (b_j, b_{j+1}, \ldots, b_n)$, $j \geq 1$, is called a
{\em subsequence} of sequence $q = (b_1, \ldots, b_n)$. We define
$q - q' = (b_1, \ldots, b_{j-1})$.

Let $Q = (q_1, \ldots, q_k)$ and $Q' = (q'_1, \ldots, q'_k)$ be two lists
of sequences of bins such that each sequence $q'_j$, $1 \leq j \leq k$,
is a subsequence of sequence $q_j$. Each such pair
$(Q,Q')$ is called a {\em configuration}. In every configuration
$(Q,Q')$ the first entry $Q$ is the initial list of sequences of bins and the
second entry $Q'$ is the list of sequences that remain to be processed. 
A pallet $t$ is called {\em open} in configuration $(Q,Q')$, if a bin of 
pallet $t$ is contained
in some $q'_i \in Q'$ and if another bin of pallet $t$ is contained in some
$q_j - q'_j$ for $q_j \in Q$, $q'_j \in Q'$. The {\em set of open pallets}
in configuration $(Q,Q')$ is denoted by $\OPEN(Q,Q')$. A pallet $t \in
\PLS(Q)$ is called {\em closed} in configuration $(Q,Q')$, if $t \not\in
\PLS(Q')$, i.e.\ no sequence of $Q'$ contains a bin for pallet $t$.
Initially all pallets are {\em unprocessed}. After the first bin of a pallet
$t$ has been removed from one of the sequences, pallet $t$ is 
either open or closed.

In view of the practical background, we only consider lists of sequences
that together contain at least two bins for each pallet. If there is only
one bin for some pallet, we can handle this bin without waiting for any
further bin. Thus, we can process instances with such special pallets
with at most one more stack-up place. Our definition of open pallets
is useless for pallets with only one bin.
Throughout the paper $Q$ always denotes a list of some sequences of bins.

\subsection*{The FIFO Stack-up Problem}

Consider a configuration $(Q,Q')$. 
The removal of the first bin from one
subsequence $q' \in Q'$ is called {\em transformation step}. A sequence
of transformation steps that transforms the list $Q$ of $k$ sequences
into $k$ empty subsequences is called a {\em processing} of $Q$.



\begin{desctight} 
\item[Name]  {\sc FIFO stack-up}


\item[Instance] A list $Q = (q_1, \ldots, q_k)$ of sequences of bins, 
for every bin $b$ of $Q$ its pallet symbol $\PL(b)$, and a positive
integer $p$.

\item[Question] Is there a processing of $Q$, such that in each configuration $(Q,Q')$ during the
processing at most $p$ pallets are open?
\end{desctight}

We use the following variables:
$k$ denotes the number of sequences, and
$p$ stands for the number of stack-up places. Furthermore,
$m$ represents the number of pallets in $\PLS(Q)$, and
$n$ denotes the total number of bins in all sequences, i.e. $n=n_k$. Finally,
$N =\max\{|q_1|, \ldots, |q_k|\}$ is the maximum sequence length.
In view of the practical background, it holds $p < m$, $k < m$, $m < n$,
and $N < n$.

It is often convenient to use pallet identifications instead of
bin identifications to represent a sequence $q$. For $r$ not
necessarily distinct pallets $t_1, \ldots, t_r$ let $[t_1,\ldots,t_r]$
denote some sequence of $r$ pairwise distinct bins $(b_1, \ldots, b_r)$,
such that $\PL(b_i) = t_i$ for $i=1,\ldots,r$.
We use this notion for lists of sequences as well. For the
sequences $q_1 = [t_1, \ldots, t_{n_1}], \ldots, q_k = [t_{n_{k-1}+1},
\ldots, t_{n_k}]$ of pallets we define $q_1 =
(b_1, \ldots, b_{n_1}), \ldots, q_k = (b_{n_{k-1}+1}, \ldots, b_{n_k})$ to
be sequences of bins such that $\PL(b_i) = t_i$ for $i = 1, \ldots, n_k$,
and all bins are pairwise distinct.

For some list of subsequences $Q'$ we define
$\FRONT(Q')$ to be the pallets of the first bins of the 
queues of $Q'$  (cf.  Example \ref{EX1} and Table \ref{TB1}).

\begin{example} \label{EX1}
Consider list $Q = (q_1, q_2)$ of sequences
$$q_1 = (b_1,\ldots,b_4) = [a,a,b,b]$$ and
$$q_2 = (b_5,\ldots,b_{12}) = [c,d,e,c,a,d,b,e].$$
Table \ref{TB1} shows a processing of $Q$ with
3 stack-up places. The underlined bin is always the bin that will be
removed in the next transformation step. We denote $Q_i = (q^i_1,q^i_2)$,
thus each row represents a configuration $(Q, Q_i)$.
\end{example}

\begin{table}[hbtp]
\[\begin{array}{|r|lll|l|c|l|}
\hline
i & q^i_1 & ~ & q^i_2 & \FRONT(Q_i)& \mbox{remove}  & \OPEN(Q,Q_i) \\
\hline
 0 & [a,a,b,b] & ~ & [\underline{c},d,e,c,a,d,b,e] &\{a,c\} & b_5 & \emptyset \\
\hline
 1 & [a,a,b,b] & ~ & [\underline{d},e,c,a,d,b,e] & \{a,d\} & b_6 & \{c\} \\
\hline
 2 & [a,a,b,b] & ~ & [\underline{e},c,a,d,b,e] & \{a,e\} &b_7 & \{c,d\} \\
\hline
 3 & [a,a,b,b] & ~ & [\underline{c},a,d,b,e] &\{a,c\} &b_8 & \{c,d,e\} \\
\hline
 4 & [\underline{a},a,b,b] & ~ & [a,d,b,e] & \{a\}&b_1 & \{d,e\} \\
\hline
 5 & [\underline{a},b,b] & ~ & [a,d,b,e] & \{a\} &b_2 & \{a,d,e\} \\
\hline
 6 & [b,b] & ~ & [\underline{a},d,b,e] & \{a,b\}&b_9 & \{a,d,e\} \\
\hline
 7 & [b,b] & ~ & [\underline{d},b,e] & \{b,d\}&b_{10} & \{d,e\} \\
\hline
 8 & [b,b] & ~ & [\underline{b},e] & \{b\}&b_{11} & \{e\} \\
\hline
 9 & [\underline{b},b] & ~ & [e] &\{b,e\} &b_3 & \{b,e\} \\
\hline
10 & [\underline{b}] & ~ & [e] & \{b,e\}&b_4 & \{b,e\} \\
\hline
11 & [] & ~ & [\underline{e}] & \{e\}&b_{12} & \{e\} \\
\hline
12 & [] & ~ & [] & \emptyset&- & \emptyset \\
\hline
\end{array}\]
\caption{A processing of $Q = (q_1, q_2)$ from Example \ref{EX1}
with 3 stack-up places. There is no processing of $Q$ that needs
less than 3 stack-up places.}
\label{TB1} 
\end{table}

Consider a processing of a list $Q$ of sequences. Let $B = (b_{\pi(1)},
\ldots, b_{\pi(n)})$ be the order in which the bins are removed during
the processing of $Q$, and let $T = (t_1, \ldots, t_m)$ be the order in
which the pallets are opened during the processing of $Q$. Then $B$ is
called a {\em bin solution} of $Q$, and $T$ is called a {\em pallet solution}
of $Q$. The transformation in Table \ref{TB1} defines the bin solution
\[
B=(b_5, b_6, b_7, b_8, b_1, b_2, b_9, b_{10}, b_{11}, b_3, b_4, b_{12}),
\]
and the pallet solution $T = (c,d,e,a,b)$.

During a processing of a list $Q$ of sequences there are often
configurations $(Q,Q')$ for which it is easy to find a bin $b$ that
can be removed from $Q'$ such that a further processing with $p$
stack-up places is still possible. This is the case, if bin $b$ is
destined for an already open pallet, see configuration $(Q,Q_3)$,
$(Q,Q_5)$, $(Q,Q_6)$, $(Q,Q_7)$, $(Q,Q_9)$, $(Q,Q_{10})$, or $(Q,Q_{11})$
in Table \ref{TB1}.
In the following we show:
\begin{itemize}
\item If one of the first bins of the sequences is destined for an already
  open pallet then this bin can be removed without increasing the number
  of stack-up places necessary to further process the sequences.
\item If there is more than one bin at choice for already open pallets
  then the order, in which those bins are removed is arbitrary.
\end{itemize}

Both rules are useful, i.e. within the following two simple algorithms
for the FIFO stack-up problem. 
\begin{description}
\item[\bf Algorithm 1] Generate all possible bin orders $(b_{\pi(1)},\ldots,b_{\pi(n)})$
and verify, whether this can be a processing. This leads to a simple but
very inefficient algorithm with running time $\bigo(n^2\cdot n!)$,
where $\bigo(n^2)$ time is needed for each test.

\item[\bf Algorithm 2] Generate all possible pallet orders $(t_{\pi(1)},\ldots,t_{\pi(m)})$
and test, whether this can be a processing. This leads to a further simple but
less inefficient algorithm with running time $\bigo(n^2\cdot m!)$.
\end{description}

The running time of Algorithm 2 is much better than those
of Algorithm 1 and Algorithm 2 is only possible because of the second rule.
If bins for already open pallets must be removed in a special order, 
Algorithm 2 would not be possible.

To show the rules, consider a processing of some list $Q$ of sequences
with $p$ stack-up places. Let 
\[
   (b_{\pi(1)}, \ldots, b_{\pi(i-1)}, b_{\pi(i)}, \ldots, b_{\pi(\ell-1)},
     b_{\pi(\ell)}, b_{\pi(\ell+1)}, \ldots, b_{\pi(n)})
\]
be the order in which the bins are removed from the sequences during the
processing, and let $(Q,Q_j)$, $1 \leq j \leq n$ denote the configuration
such that bin $b_{\pi(j)}$ is removed in the next transformation step.
Suppose bin $b_{\pi(i)}$ will be removed in some transformation step
although bin $b_{\pi(\ell)}$, $\ell > i$, for some already open pallet
$\PL(b_{\pi(\ell)}) \in \OPEN(Q,Q_i)$ could be removed next. We define
a modified processing
\[
   (b_{\pi(1)}, \ldots, b_{\pi(i-1)}, b_{\pi(\ell)}, b_{\pi(i)}, \ldots,
    b_{\pi(\ell-1)}, b_{\pi(\ell+1)}, \ldots, b_{\pi(n)})
\]
by first removing bin $b_{\pi(\ell)}$, and afterwards the bins
$b_{\pi(i)}, \ldots, b_{\pi(\ell-1)}$ in the given order.
Obviously, in each configuration during the modified processing
there are at most $p$ pallets open. To remove first some bin of an already
open pallet is a kind of priority rule.

A configuration $(Q, Q')$ is called a {\em decision configuration}, if the
first bin of each sequence $q' \in Q'$ is destined for a non-open pallet,
see configurations $(Q,Q_0)$, $(Q,Q_1)$, $(Q,Q_2)$, $(Q,Q_4)$, and
$(Q,Q_8)$ in Table \ref{TB1}, i.e. $\FRONT(Q')\cap \OPEN(Q,Q')=\emptyset$. 
We can restrict FIFO stack-up algorithms to deal with such decision
configurations, in all other configurations the algorithms automatically
remove a bin for some already open pallet.

If we have a pallet solution computed by some FIFO stack-up algorithm, we
need to convert the pallet solution into a sequence of transformation
steps, i.e.\ a processing of $Q$. This is done by algorithm {\sc Transform}
in Figure \ref{fig:algorithm1}.
Given a list of sequences $Q = (q_1, \ldots, q_k)$ and
a pallet solution $T = (t_1, \ldots, t_m)$ algorithm {\sc Transform} gives
us in time $\bigo(n\cdot k)\subseteq \bigo(n^2)$ 
a bin solution of $Q$, i.e. a processing of $Q$.

\begin{figure}[h]
\hrule
{\strut\footnotesize \bf Algorithm {\sc Transform}} 
\hrule
\medskip
\medskip
$q'_1 := q_1, \ldots, q'_k := q_k$ \\
$j := 1$ \\
$T' := \{t_1\}$ \\[2pt]
repeat the following steps until $q'_1 = \emptyset, \ldots, q'_k =
\emptyset$:
\begin{enumerate}
\item if there is a sequence $q'_i$ such that the first bin $b$ of
  $q'_i$ is destined for a pallet in $T'$, i.e.\ $\PL(b) \in T'$,
  then remove bin $b$ from sequence $q'_i$, and output $b$

  {\it Comment:} Bins for already open pallets are removed automatically.
\item otherwise set $j := j+1$ and $T' := T' \cup \{t_j\}$

  {\it Comment:} If the first bin of each subsequence $q'_i$ is destined
  for a non-open pallet, the next pallet of the pallet solution has to be
  opened.
\end{enumerate}

\normalsize
\hrule
\caption{Algorithm for transforming a pallet solution into a bin solution.}
\label{fig:algorithm1}
\end{figure}

Obviously, there is no other processing of $Q$ that also defines pallet
solution $T$ but takes less stack-up places.

\section{Digraph Models} \label{SC-dm}

Next we introduce two digraph models 
for the FIFO stack-up problem.
To solve an instance of the FIFO stack-up problem we transform it into a digraph 
and by solving a graph problem we obtain a solution for the FIFO stack-up problem.

\subsection{The processing graph} \label{SC2}

In this section we give an algorithm which computes a processing
of the given sequences of bins with a minimum number of stack-up places.
Such an optimal processing can always be found by computing the {\it
processing graph} and doing some calculation on it. Before we define the
processing graph let us consider some general graph problem, that will be
useful to find a processing
of the given sequences of bins with a minimum number of stack-up places and 
thus to solve the FIFO stack-up problem.

Let $G=(V,E,f)$ be a directed acyclic vertex-labeled graph. Function
$f: V \to \ZZ$ assigns to every vertex $v \in V$  a value $f(v)$. Let
$s \in V$ and $t \in V$ be two vertices, where $s$ stands for source, and
$t$ stands for target. For some vertex $v \in V$ and some
path $P=(v_1, \ldots, v_\ell)$ with $v_1 = s$, $v_\ell = v$ and
$(v_i, v_{i+1}) \in E$ we define
\[
   val(P) := \max_{u \in P} (f(u))
\]
to be the maximum value of $f$ on that path. Let ${\mathcal P}_s(v)$ 
denote the set of all paths from vertex $s$ to vertex $v$. Then we define
\[
   val(v) := \min_{P \in {\mathcal P}_s(v)} (val(P)).
\]
If $v$ is not reachable from $s$, we define $val(v)=\infty$.
The problem is to compute the value $val(t)$. 
A solution of this problem can be found by dynamic programming and solves
also the FIFO stack-up problem. The reason for this is given shortly.
For some vertex $v \in V$ let
\[
   N^{-}(v) := \{ u \in V \mid (u,v) \in E\}
\]
be the set of predecessors of $v$ in digraph $G$. The acronym $N^{-}$ stands
for in-neighborhood. Then it holds
\[
   val(v) = \max\{ f(v), \min_{u \in N^{-}(v)}(val(u)) \},
\]
because each path $P \in {\mathcal P}_s(v)$ must go through some vertex $u\in N^{-}(v)$. 
That means, the value of $val(v)$ can be processed recursively. But, if
we would do that, subproblems often would be solved several times. So,
we use dynamic programming to solve each subproblem only once, and put
these solutions together to a solution of the original problem. This is
possible, since the graph is directed and acyclic.

Let $topol: V \to \NN$ be an ordering of the vertices of digraph $G$ such that
$topol(u) < topol(v)$ holds for each $(u,v) \in E$. That means, $topol$ is
a topological ordering of the vertices.

\begin{figure}[ht]
\hrule
{\strut\footnotesize \bf Algorithm {\sc Opt}} 
\hrule
\smallskip 
\begin{tabbing}
xxxx \= xxxx \= xxxx \= xxxx \=xxxx \= xxxx \=xxxx \= xxxx \=\kill
$val[s] := f(s)$ \\
for every vertex $v \neq s$ in order of $topol$ do \\
\> $val[v] := \infty$ \\
\> for every $u \in N^{-}(v)$ do \>\>\>\>\>\>  $\vartriangleright$ Compute $\min_{u \in N^{-}(v)}(val(u))=:val(v)$ \\
\>\> if ($val[u] < val[v]$) \\
\>\>\> $val[v] := val[u]$ \\
\>\>\> $pred[v] := u$ \\
\> if ($val[v] < f(v)$) \>\>\>\>\>\>  $\vartriangleright$ Compute  $\max\{ f(v), val(v)\}$ \\
\>\> $val[v] := f(v)$
\end{tabbing}

\hrule
\caption{Finding an optimal processing by dynamic programming, where $topol$
is a topological ordering of the vertices.}
\label{fig:algorithm2}
\end{figure}

The value $val(t)$ can be computed in polynomial time by Algorithm {\sc Opt}
given in Figure \ref{fig:algorithm2}. A topological ordering of the vertices
of digraph $G$ can be found by a depth first search algorithm in time
$\bigo(|V| + |E|)$. The remaining work of algorithm {\sc Opt} can also
be done in time $\bigo(|V| + |E|)$. Thus, it works in time
$\bigo(|V| + |E|)$, i.e. linear in 
the input 
size.

\begin{remark}
An alternative 
algorithm to compute the values $val(v)$ for $v\in V$ in $G=(V,E,f)$
was suggested by an anonymous reviewer. 
For every $v\in V$ the value  $val(v)$ can be computed as follows.
Start with $r=\max_{u\in V}f(u)$. While the graph contains at least
one path from $s$ to $v$, we remove all vertices $u$ with $f(u)\geq r$
and decrease $r$ by one. Finally return $val(v)=r+1$.
The running time is not better but it shows immediately that it can be 
done in polynomial time. Independently we introduced a  
breadth first search solution for the FIFO stack-up problem in \cite{GRW15a}.
\end{remark}

We need some additional terms, before we show how algorithm {\sc Opt} 
can solve the FIFO stack-up problem. For a sequence $q = (b_1, \ldots, b_n)$ 
let
\[
   \links(q, i) := (b_1, \ldots, b_i)
\]
denote the sequence of the first $i$ bins of sequence $q$, and let
\[
   \rechts(q, i) := (b_{i+1}, \ldots, b_n)
\]
denote the remaining bins of sequence $q$ after removing the first $i$ bins.
Consider again Example \ref{EX1} and Table \ref{TB1}. It can be seen that
a configuration is well-defined by the number of bins that are removed from
each sequence. Configuration $(Q,Q_6)$ can be described equivalently by the
tuple $(2,4)$, since in configuration $(Q,Q_6)$ two bins of sequence $q_1$
have been removed, and four bins of sequence $q_2$ have been removed.

The position of the first
bin in some sequence $q_i$ destined for some pallet $t$ is denoted by
$\FI(q_i,t)$, similarly the position of the last bin for pallet $t$ in
sequence $q_i$ is denoted by $\LA(q_i,t)$. For technical reasons, if 
there is no bin for pallet $t$
contained in sequence $q_i$, then we define $\FI(q_i,t) = |q_i|+1$, and
$\LA(q_i,t) = 0$.\label{label-page-fi-and-la}

\begin{example} \label{EX5}
Consider list $Q = (q_1, q_2)$ of the sequences $$q_1 = [a,b,c,a,b,c]$$
and $$q_2 = [d,e,f,d,e,f,a,b,c].$$ Then we get $\links(q_1,2) = [a,b]$,
$\rechts(q_1,2) = [c,a,b,c]$, $\links(q_2,3) = [d,e,f]$, and
$\rechts(q_2,3) = [d,e,f,a,b,c]$.

If we denote $q'_1 := \rechts(q_1,2)$, and $q'_2 := \rechts(q_2,3)$,
then there are 5 pallets open in configuration
$(Q,Q')$ with $Q' = (q'_1, q'_2)$, namely the pallets
$a$, $b$, $d$, $e$, and $f$. 
\end{example}

We generalize this for a list $Q = (q_1, \ldots, q_k)$ of sequences 
and we define the cut
\begin{eqnarray*}
  \lefteqn{\CUT_Q(i_1, \ldots, i_k) ~ := ~ \{ t \in \PLS(Q) \mid } \\
  & ~~ & \exists j,j', b \in \links(q_j, i_j),
       b' \in \rechts(q_{j'}, i_{j'}): \PL(b) = \PL(b') = t \}
\end{eqnarray*}
at some configuration $(i_1, \ldots, i_k)$
to be the set of pallets $t$ such that one bin for pallet $t$ has already
been removed and another bin for pallet $t$ is still contained in some
sequence. Let $\# \CUT_Q(i_1, \ldots, i_k)$ be the number of elements
in $\CUT_Q(i_1, \ldots, i_k)$.

The intention of a processing graph $G=(V,E,f)$ is the following. Suppose each
vertex $v \in V$ of digraph $G$ represents a configuration $(i_1, \ldots, i_k)$
during the processing of some set of sequences $Q$. Suppose further,
an arc $(u,v) \in E$ represents a transformation step during this
processing, such that a bin $b$ is removed from some sequence in
configuration $u$ resulting in configuration $v$. Suppose also that
each vertex $v$ is assigned the number of open pallets in configuration
$v$, i.e.\ number $\# \CUT_Q(i_1, \ldots, i_k)$.
If vertex $s$ represents the initial configuration
$(0,0, \ldots, 0)$, while vertex $t$ represents the final configuration
$(|q_1|, |q_2|, \ldots, |q_k|)$, then we are searching a path $P$ from
$s$ to $t$ such that the maximal number on path $P$ is minimal. Thus,
an optimal processing of $Q$ can be found by Algorithm {\sc Opt} given in
Figure \ref{fig:algorithm2}.

The {\em processing graph} has a vertex for each possible configuration. 
Each vertex $v$ of the processing graph is labeled by the vector
$h(v) = (v_1, \ldots, v_k)$, where $v_i$ denotes
the position of the bin that has been
removed last from sequence $q_i$.
Further  each vertex $v$ of the processing graph is labeled by
the value $f(v)=\# \CUT_Q(v_1, \ldots, v_k)=\# \CUT_Q(h(v))$.
There is an arc from vertex $u$
labeled by $(u_1, \ldots, u_k)$ to vertex $v$ labeled by $(v_1, \ldots, v_k)$ if
and only if $u_i = v_i - 1$ for exactly one element of the vector
and for all other elements of the vector
$u_j = v_j$. The arc is labeled with the pallet
symbol of that bin, that will be removed in the corresponding
transformation step. That means, the processing graph
represents all the possible branchings through the decision process.
For the sequences of Example \ref{EX5}
we get the processing graph of Figure \ref{F06}. The processing graph is
directed and acyclic, and we use this digraph to compute the values of
$\# \CUT_Q(i_1, \ldots, i_k)$ iteratively in the following way.

\begin{figure}[hbtp]
\centerline{\epsfxsize=75mm \epsfbox{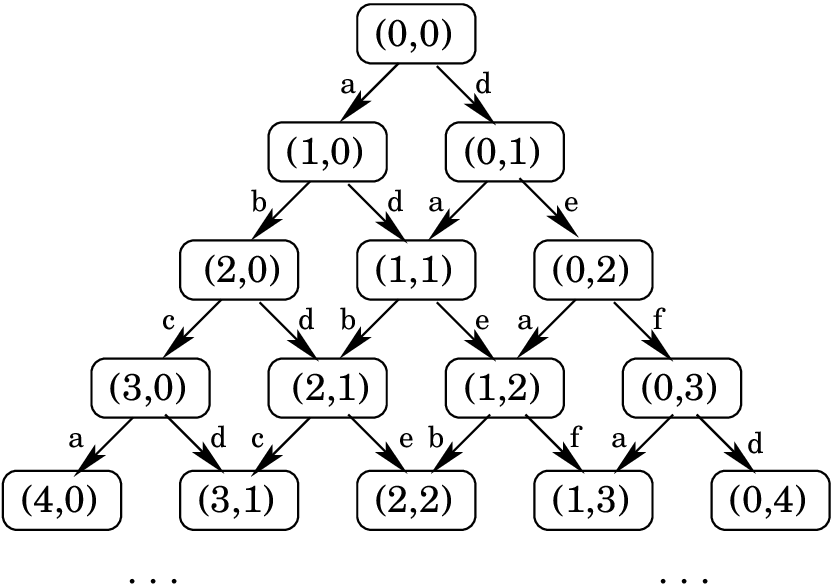}}
\caption{The processing graph of Example \ref{EX5}.
For every
vertex $v$ the vector $h(v)$ in the vertex represents the configuration of $v$ and  
every arc is labeled with the pallet
symbol of that bin, that is removed in the corresponding
transformation step. 
}
\label{F06}
\end{figure}

First, since none of the bins has been removed from some sequence,
we have $\# \CUT_Q(0, \ldots, 0) = 0$. Since the processing graph is
directed and acyclic, there exists a topological ordering $topol$ of the
vertices. The vertices are processed according to the order
$topol$. In each
transformation step we remove exactly one bin for some pallet $t$ from
some sequence $q_j$, thus
\begin{eqnarray*}
  \lefteqn{\# \CUT_Q(i_1, \ldots, i_{j-1}, i_j+1, i_{j+1}, \ldots, i_k)} \\
  & ~~ = &
   \# \CUT_Q(i_1, \ldots, i_{j-1}, i_j, i_{j+1}, \ldots, i_k) ~+~ c_{j}
\end{eqnarray*}
where
$c_{j} = 1$ if pallet $t$ has been opened in the
transformation step, and $c_{j} = -1$ if pallet $t$ has been closed
in the transformation step. Otherwise, $c_{j}$ is zero. If we
put this into a formula we get
\[
   c_{j} = \left\{
     \begin{array}{rll}
       1, & \mbox{if } \FI(q_j,t) = i_j+1 \mbox{ and }   \FI(q_\ell,t) > i_\ell ~\forall ~ \ell \neq j \\
      -1, & \mbox{if } \LA(q_j,t) = i_j+1 \mbox{ and }  \LA(q_\ell,t) \leq i_\ell ~ \forall ~ \ell \neq j \\
       0, & \mbox{otherwise.}
     \end{array}
   \right.
\]

Please remember our technical definition of $\FI(q,t)$ and $\LA(q,t)$
from page \pageref{label-page-fi-and-la}
for the case that $t\not\in \PLS(q)$.

That means, the calculation of value $\# \CUT_Q(i_1, \ldots, i_k)$ for the vertex
labeled $(i_1, \ldots, i_k)$ depends only on already calculated values.
Figure \ref{F07} shows such a processing for the sequences of Example
\ref{EX5}.

\begin{figure}[hbtp]
\centerline{\epsfxsize=83mm \epsfbox{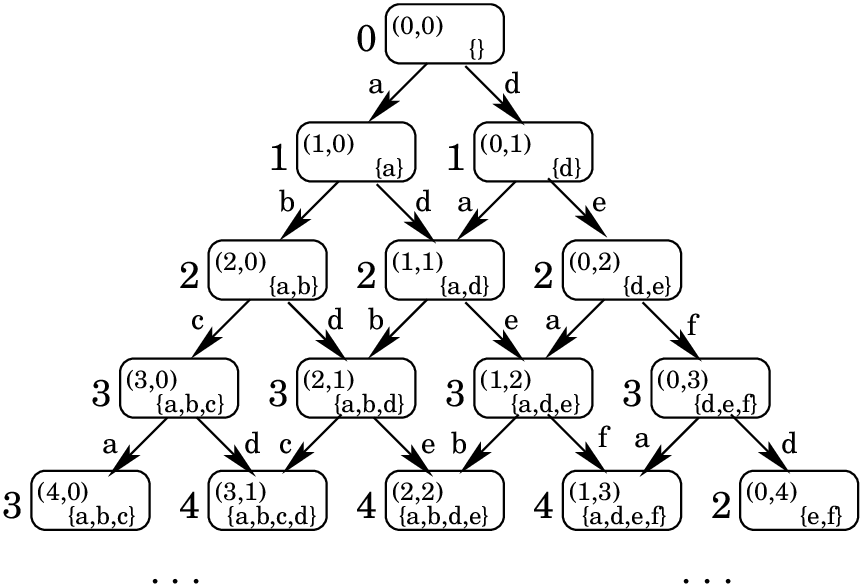}}
\caption{The processing graph of Example \ref{EX5}. For every
vertex $v$ the set in the vertex $v$ shows the open pallets, the
vector $h(v)$ in the vertex represents the configuration of $v$, and  
the number $f(v)$ on the left of the vertex represents the number of open pallets 
in the corresponding configuration $h(v)$.}
\label{F07}
\end{figure}

To efficiently perform the processing of $c_j$, we have to store for each pallet
the first and last bin in each sequence. Table \ref{TB3} shows these values
for the sequences of Example
\ref{EX5}. We compute all the values $\FI(q_j,t)$ and $\LA(q_j,t)$ for $q_i\in Q$ and
$t\in \PLS(Q)$ in some preprocessing phase. 
Such a preprocessing can be done in time $\bigo(k \cdot N+k\cdot m)$, which can
be bounded due to $m\leq k\cdot N$ by $\bigo(k\cdot (N+1)^k)$. 
Then the calculation of value $\# \CUT_Q(h(v))$ for the
vertex $v$ representing configuration $h(v)$ can be done in
time $\bigo(k)$.

%
%


\begin{table}[hbt]
\centering
\begin{tabular}{|r||c|c|c|c|c|c|}
\hline
pallet $t$ & ~$a$~ & ~$b$~ & ~$c$~ & ~$d$~ & ~$e$~ & ~$f$~ \\
\hline
\hline
$\FI(q_1, t)$ & 1 & 2 & 3 & 7 & 7 & 7 \\
$\FI(q_2, t)$ & 7 & 8 & 9 & 1 & 2 & 3 \\
\hline
\hline
$\LA(q_1, t)$ & 4 & 5 & 6 & 0 & 0 & 0 \\
$\LA(q_2, t)$ & 7 & 8 & 9 & 4 & 5 & 6 \\
\hline
\end{tabular}
\caption{The values $\FI(q_j,t)$ and $\LA(q_j,t)$ for $q_i\in Q$ and
$t\in \PLS(Q)$ of Example \ref{EX5}.}
\label{TB3}
\end{table}

The computation of all at most $(N+1)^k$ values $\# \CUT_Q(i_1, \ldots, i_k)$ 
can be performed in time $\bigo(k \cdot (N+1)^k)$.
After that we can use Algorithm {\sc Opt} to compute the minimal 
number of stack-up places necessary to process the given FIFO stack-up 
problem in time $\bigo(|V| + |E|)\subseteq \bigo(k \cdot (N+1)^k)$, since $|V|\leq (N+1)^k$ and $|E|\leq k\cdot |V|$. 
If the size of the processing graph is polynomially bounded in the size
of the input, the FIFO stack-up problem
can be solved in polynomial time. In general the processing graph 
can be too large to efficiently solve the problem.


\subsection{The Sequence Graph} \label{seq-gr}

Next we show a correlation 
between the used number of stack-up places for a processing of an instance $Q$ 
and the directed pathwidth of a digraph $G_Q$ defined by $Q$. The notion 
of {\em directed pathwidth} was introduced by Reed, Seymour, and Thomas in the 1990s
and leads a restriction of {\em directed treewidth} which is defined by
Johnson, Robertson, Seymour, and Thomas in \cite{JRST01}.

This correlation implies that \begin{inparaenum}[1.)]
\item  the 
FIFO stack-up problem is NP-complete and 
\item a pallet solution can be computed in 
polynomial time if there are only a fixed number of stack-up places.
\end{inparaenum}

A {\em directed path-decomposition} of a digraph $G=(V,E)$ is a 
sequence $$(X_1, \ldots, X_r)$$ of subsets of $V$, called {\em bags}, that satisfy the 
following three properties.

\begin{description}
\item[\bf (dpw-1)] $X_1 \cup \ldots \cup X_r ~=~ V$ 
\item[\bf (dpw-2)] for each arc $(u,v) \in E$ there are indices $i,j$ 
with $i \leq j$ 
such that $u \in X_i$ and $v \in X_j$ 
\item[\bf (dpw-3)] if $u \in X_i$ and $u \in X_j$ for some node $u$ and two 
indices $i,j$ with $i \leq j$, then $u \in X_\ell$ for all indices $\ell$ 
with $i \leq \ell \leq j$
\end{description}

The {\em width} of a directed path-decomposition ${\cal X}=(X_1, \ldots, X_r)$ 
is $$\max_{1 \leq i \leq r} |X_i|-1.$$ The {\em directed pathwidth} of $G$,
$\dpw(G)$ for short, is 
the smallest integer $w$ such that there is a directed path-decomposition for 
$G$ of width $w$. For symmetric digraphs, the directed pathwidth is equivalent 
to the undirected pathwidth of the corresponding undirected graph \cite{KKKTT12}, which
implies that determining
whether the pathwidth of some given digraph  is 
at most some given value $w$ is NP-complete. 
For each fixed integer $w$, it is decidable in polynomial time whether a given 
digraph has directed pathwidth at most $w$, see Tamaki \cite{Tam11}.

The {\em sequence graph} $G_Q = (V,E)$ for an instance $Q=(q_1,\ldots,q_k)$ 
of the FIFO stack-up problem is defined by vertex set $V= \PLS(Q)$ and 
the following set of arcs. 
There is an arc $(u,v) \in E$ if and only if there is some sequence 
where a bin destinated for $u$ is on the left of some bin destinated for $v$.

More formally, there is
an arc $(u,v) \in E$ if and only if 
there is a sequence $q_i=(b_{n_{i-1}+1},\ldots,b_{n_i})$ with two bins 
$b_{j_1},b_{j_2}$ such that
\begin{inparaenum}[(1)]
\item $j_1 < j_2$,
\item $\PL(b_{j_1})=u$,
\item $\PL(b_{j_2})=v$, and
\item $u \not= v$.
\end{inparaenum}

\begin{example} \label{EX6}
Figure \ref{F04} shows the sequence graph $G_Q$ for $Q = (q_1, q_2, q_3)$ 
with sequences $q_1 = [a,a,d,e,d]$, $q_2 = [b,b,d]$, and 
$q_3 = [c,c,d,e,d]$.
\end{example}

\begin{figure}[hbt]
\centerline{\epsfxsize=45mm \epsfbox{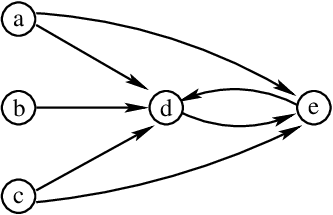}}
\caption{Sequence graph $G_Q$ of Example \ref{EX6}.}
\label{F04}
\end{figure}

If $G_Q = (V,E)$ has an arc $(u,v) \in E$ then $u \not= v$ and for 
every processing of $Q$, 
pallet $u$ is opened before pallet $v$ is closed. Digraph $G_Q = (V,E)$ can 
be computed in time 
$\bigo(n + k\cdot |E|)\subseteq \bigo(n + k\cdot  m^2)$ by the algorithm
{\sc Create Sequence Graph} shown in Figure \ref{fig:algorithm5}.

\begin{figure}[ht]
\hrule
{\strut\footnotesize \bf Algorithm {\sc Create Sequence Graph}} 
\hrule
\begin{tabbing}
xxxx \= xxxx \= xxxx \= xxxx \= xxxx \= xxxx \= xxxx \=\kill
for each sequence $q \in Q$ do \\
\> $b :=$ first bin of sequence $q$ \\
\> add $\PL(b)$ to vertex set $V$, if it is not already contained \\
\> $L := (\PL(b))$ \>\>\>\>\>\>   $\vartriangleright$ $L$ contains pallets of bins up to bin $b$ \\
\> for $i := 2$ to $|q|$ do \\
\>\> $b := i$-th bin of sequence $q$ \\
\>\> add $\PL(b)$ to vertex set $V$, if it is not already contained \\
\>\> if ($i== \LA(q,\PL(b))$)\\
\>\>\> for each pallet $t \in L$ do \\
\>\>\>\>  if $t \neq \PL(b)$ add arc $(t,\PL(b))$ to arc set $A$, if it is not already contained \\
\>\> if ($i== \FI(q,\PL(b))$)\\
\>\>\>  $append(\PL(b),L)$ 
\end{tabbing}
\hrule
\caption{Create the sequence graph $G=(V,A)$ for some given list of sequences $Q$.}
\label{fig:algorithm5}
\end{figure}

A value is added to a list only if it is not already contained. To check
this efficiently in time $\bigo(1)$ we have to use a boolean array. 
In our algorithm $V$ and $L$ are realized by boolean arrays.
Therefore we need some preprocessing phase where we run through each sequence and seek
for the pallets. This can be done in time $\bigo(n+k\cdot m)\subseteq \bigo(n+m^2)$.

\medskip
The following  two Theorems show the correlation 
between the used number of stack-up places for a processing of an instance $Q$ 
and the directed pathwidth of the sequence graph $G_Q$. 

First we want to emphasize
that not every directed path-decomposition of a sequence graph $G_Q$
immediately leads to a pallet solution. In our Example \ref{EX6} the
sequence $(\{e,a\},\{e,b\},\{e,c\},\{e,d\})$ 
is a directed path-decomposition of optimal width $1$ 
for the sequence graph $G_Q$.
But opening the pallets one after another leads to $(e,a,b,c,d)$, which is no pallet solution 
since pallet $e$ cannot be opened at first and must be put on hold. 
Within the proof of  Theorem \ref{th-pd-q} we show
how to transform a directed path-decomposition of $G_Q$ into a 
pallet solution for $Q$. Example  \ref{EX7} and Table \ref{TB2}
illustrate this process.

\begin{theorem}\label{th-q-pd}
A processing $(Q,Q_0),(Q,Q_1),\ldots,(Q,Q_n)$ of $Q$ with $Q_0=Q$, 
$Q_n=(\emptyset,\ldots,\emptyset)$, 
and $p$ stack-up places defines a directed path-decomposition 
$${\cal X}=(\OPEN(Q,Q_0),\ldots,\OPEN(Q,Q_n))$$ for $G_Q$ of width $p-1$.
\end{theorem}

\begin{proof}
We show that $${\cal X}=(\OPEN(Q,Q_0),\ldots,\OPEN(Q,Q_n))$$ satisfies 
all properties of a directed path-decomposition.
\begin{description}
\item[(dpw-1)]
$\OPEN(Q,Q_0) \cup \cdots \cup \OPEN(Q,Q_n) = \PLS(Q)$,
because every pallet is opened at least once during a processing.
\item[(dpw-2)]
If $(u,v) \in E$ then there are indices $i,j$ with $i \leq j$ such 
that $u \in \OPEN(Q,Q_i)$ and $v \in \OPEN(Q,Q_j)$, because $v$ can 
not be closed before $u$ is opened.  
\item[(dpw-3)]
If $u \in \OPEN(Q,Q_i)$ and $u \in \OPEN(Q,Q_j)$ for some pallet $u$ and 
two indices $i,j$ with $i \leq j$, then $u \in \OPEN(Q,Q_\ell)$ for all 
indices $\ell$ with $i \leq \ell \leq j$, because every pallet is opened at 
most once.
\end{description}
Since $Q$ is processed with $p$ stack-up places, we have $|\OPEN(Q,Q_i)| \leq p$ 
for $0 \leq i \leq n$, and therefore ${\cal X}$
has width at most $p-1$.
\end{proof}

\begin{theorem}\label{th-pd-q}
If there is a path-decomposition ${\cal X}=(X_1, \ldots, X_r)$ for $G_Q$ of 
width $p-1$ then there is a processing of $Q$ with $p$ stack-up places.
\end{theorem}

\begin{proof}
For a pallet $t$ let $\alpha({\cal X},t)$ be the smallest $i$ such that 
$t \in X_i$ and $\beta({\cal X},t)$ be the largest $i$ such that $t \in X_i$,
see Example \ref{EX7}.
Then $t \in X_i$ if and only if $\alpha({\cal X},t)\le i \le \beta({\cal X},t)$. 
If $(t_1,t_2)$ is an arc of $G_Q$, then $\alpha({\cal X},t_1) \leq \beta({\cal X},t_2)$. 
This follows by (dpw-2) of the definition of a directed path-decomposition.

Instance $Q$ can be processed as follows. If it is necessary to open in a 
configuration $(Q,Q')$ a new pallet, then we open a pallet $t$ of $\FRONT(Q')$ 
for which $\alpha({\cal X},t)$ is minimal.

We next show that for every configuration $(Q,Q')$ of the processing above 
there is a bag $X_i$ in the directed path-decomposition ${\cal X}=(X_1, \ldots, X_r)$, such 
that $\OPEN(Q,Q') \subseteq X_i$. This implies that the processing uses 
at most $p$ stack-up places.

First, let $(Q,Q')$ be a decision configuration, let $t \in \FRONT(Q')$ such that 
$\alpha({\cal X},t)$ is minimal, and let $t' \in \OPEN(Q,Q')$ be an 
already open pallet.
We show that the intervals $[\alpha({\cal X},t),\beta({\cal X},t)]$
and $[\alpha({\cal X},t'),\beta({\cal X},t')]$ overlap, because
$\alpha({\cal X},t) \leq \beta({\cal X},t')$ and 
$\alpha({\cal X},t') \leq \beta({\cal X},t)$.

\begin{enumerate}[(1)]
\item To show $\alpha({\cal X},t) \leq \beta({\cal X},t')$ 
we observe the following:

Since $t' \not\in \FRONT(Q')$, there has to be a pallet 
$t'' \in \FRONT(Q')$ such that $(t'',t')$ is an arc of $G_Q$. 
This implies that $\alpha({\cal X},t'') \leq \beta({\cal X},t')$. 
Since $t$ is a pallet of $\FRONT(Q')$ for which $\alpha({\cal X},t)$ 
is minimal, we have $\alpha({\cal X},t) \leq \alpha({\cal X},t'')$ and 
thus $\alpha({\cal X},t) \leq \beta({\cal X},t')$. 

\item To show $\alpha({\cal X},t') \leq \beta({\cal X},t)$ 
we observe the following:

Let $(Q,Q'')$ be the configuration in which $t'$ has been opened.
\begin{enumerate}
\item $t \in \FRONT(Q'')$. Since $t'$ is a pallet of $\FRONT(Q'')$ 
for which $\alpha({\cal X},t')$ is minimal, we have 
$\alpha({\cal X},t') \leq \alpha({\cal X},t)$ and thus 
$\alpha({\cal X},t') \leq \beta({\cal X},t)$.
\item $t \not\in \FRONT(Q'')$. Then there is a palled 
$t'' \in \FRONT(Q'')$ such that $(t'',t)$ is an arc of $G_Q$. 
This implies that $\alpha({\cal X},t'') \leq \beta({\cal X},t)$. 
Since $t'$ is a pallet of $\FRONT(Q'')$ for which $\alpha({\cal X},t')$ 
is minimal, we have $\alpha({\cal X},t') \leq \alpha({\cal X},t'') \leq \beta({\cal X},t)$.
\end{enumerate}
\end{enumerate}

%
%

%
%
Finally, let $(Q,\hat{Q})$ be an arbitrary configuration during 
the processing of $Q$.
By the discussion above, we can conclude that for every pair $t_i,t_j\in \OPEN(Q,\hat{Q})$
the intervals $[\alpha({\cal X},t_i),\beta({\cal X},t_i)]$
and $[\alpha({\cal X},t_j),\beta({\cal X},t_j)]$
overlap, because $t_i$ is opened before $t_j$ or vice versa.
Since the cut of all intervals $[\alpha({\cal X},t_i),\beta({\cal X},t_i)]$,
$t_i\in \OPEN(Q,\hat{Q})$ mutually overlap,
there is a bag $X_j$ 
in the directed path-decomposition ${\cal X}$ such that $\OPEN(Q,\hat{Q}) \subseteq X_j$.
%
%
%
%
\end{proof}

\begin{example} \label{EX7}
We consider digraph $G_Q$ for $Q = (q_1, q_2, q_3)$ 
with sequences $q_1 = [a,a,d,e,d]$, $q_2 = [b,b,d]$, and 
$q_3 = [c,c,d,e,d]$ from Example \ref{EX6}. 
Sequence $(\{a,e\}$, $\{b,e\}$, $\{c,e\}$, $\{d,e\})$ 
is a directed path-decomposition of width $1$,
which implies the following values for $\alpha$
and $\beta$ used in the proof of Theorem \ref{th-pd-q}.
$$
\begin{array}{|l|c|c|c|c|c|}
\hline
\text{pallet } t  & ~a~ & ~b~ & ~c~ & ~d~ & ~e~ \\
\hline
\hline
\alpha({\cal X},t) & 1 & 2 & 3 & 4 & 1 \\
\hline
\beta({\cal X},t)  & 1 & 2&  3 & 4 & 4  \\
\hline
\end{array}
$$
Table \ref{TB2} 
shows a processing of $Q$ with $2$ stack-up places and 
pallet solution $S=(a,b,c,d,e)$.  
The underlined bin is always the bin that will be
removed in the next transformation step. We denote $Q_i = (q^i_1,q^i_2,q^i_3)$,
thus each row represents a configuration $(Q, Q_i)$.
\end{example}

\begin{table}[hbtp]
\[\begin{array}{|r|lllll|l|l|}
\hline
 i & q^i_1 & ~ & q^i_2 &  ~ & q^i_3  & \FRONT(Q')&\OPEN(Q,Q_i)    \\
\hline
 0 & [\underline{a},a,d,e,d] & ~ & [b,b,d] & ~ & [c,c,d,e,d] &  \{a,b,c\}& \emptyset  \\
\hline
 1 & [\underline{a},d,e,d] & ~ & [b,b,d] & ~ & [c,c,d,e,d]  & \{a,b,c\} &\{a\} \\
\hline
 2  & [d,e,d] & ~ & [\underline{b},b,d] & ~ & [c,c,d,e,d]  &  \{b,c,d\}&\emptyset\\
\hline
 3   & [d,e,d] & ~ & [\underline{b},d] & ~ & [c,c,d,e,d]  & \{b,c,d\} &\{b\} \\
\hline
 4  & [d,e,d] & ~ & [d] & ~ & [\underline{c},c,d,e,d]  & \{c,d\} &\emptyset\\
\hline
 5  & [d,e,d] & ~ & [d] & ~ & [\underline{c},d,e,d]  &  \{c,d\}&\{c\}\\
\hline
 6  & [\underline{d},e,d] & ~ & [d] & ~ & [d,e,d]   &  \{d\}&\emptyset\\
\hline
 7   & [e,d] & ~ & [\underline{d}] & ~ & [d,e,d]  & \{d,e\} &\{d\} \\
\hline
 8 & [e,d] & ~ & [] & ~ & [\underline{d},e,d]   & \{d,e\}&\{d\} \\
\hline
 9& [\underline{e},d] & ~ & [] & ~ & [e,d]  &\{e\} &\{d\} \\
\hline
 10 & [\underline{d}] & ~ & [] & ~ & [e,d]  & \{d,e\}& \{d,e\}\\
\hline
 11  & [] & ~ & [] & ~ & [\underline{e},d] & \{e\}& \{d,e\}\\
\hline
 12 & [] & ~ & [] & ~ & [\underline{d}]   & \{d\}& \{d\} \\
\hline
 13 & [] & ~ & [] & ~ & []  &  \emptyset & \emptyset \\
\hline
\end{array}\]
\caption{A processing of $Q$ with respect to a given directed path-de\-com\-po\-sition 
for $G_Q$ of Example \ref{EX7}.}
\label{TB2} 
\end{table}

\section{Hardness Result}

Next we will show the hardness of the FIFO stack-up problem. In contrast to 
Section \ref{SC-dm} we will transform an instance of a graph problem into 
an instance of the FIFO stack-up problem.

Let $G=(V,E)$ be a digraph. We will assume that $G=(V,E)$
does not contain any vertex with only  outgoing arcs
and not contain any vertex with only incoming arcs.
This is only for technical 
reasons and the removal of such vertices 
will not change the directed pathwidth of $G$, because 
a vertex $u$ with only outgoing arcs can be 
placed in a singleton $X_i=\{u\}$ at the beginning
of the directed path-decomposition and a vertex $u$ with only incoming arcs can be 
placed in a singleton $X_i=\{u\}$ at the end
of the directed path-decomposition, without to change its width.

Let $G=(V,E)$ be some  digraph and $E=\{e_1,\ldots,e_\ell\}$ its arc set.
The {\em queue system} $Q_G = (q_1,\ldots,q_\ell)$ for $G$ is defined as follows.
\begin{enumerate}[(1)]
\item There are $2\ell$ bins $b_1,\ldots,b_{2\ell}$.

\item Queue $q_i=(b_{2i-1},b_{2i})$ for $1\leq i \leq \ell$.

\item The pallet symbol of bin $b_{2i-1}$ is the first vertex of arc $e_i$
and the pallet symbol of $b_{2i}$ is the second vertex of arc $e_i$ for
$1\leq i \leq \ell$. Thus  $\PLS(Q_G) = V$.
\end{enumerate}

The definition of queue system $Q_G$ and sequence graph 
$G_Q$, defined in Section \ref{seq-gr}, now imply the following proposition.

\begin{proposition}\label{prop}
For every digraph $G$,
$$G \quad = \quad G_{Q_G}.$$
\end{proposition}

\begin{example}\label{EX3-new}
Consider the  digraph $G$ of Figure \ref{F-ex-d}.
The corresponding queue system is $Q_G=(q_1,q_2,q_3,q_4,q_5,q_6,q_7)$,
where
$$
\begin{array}{lcllcllcllcl}
q_1 &= &[a,b], & q_2 &=& [b,c], & q_3 &=& [c,d], &q_4 &=& [d,e],  \\
q_5 &= &[e,a], & q_6 &=& [e,f], & q_7 &=& [f,a].
\end{array}$$
The sequence graph of $Q_G$ is $G$.
\end{example}

\begin{figure}[hbtp]
\centerline{\epsfxsize=50mm \epsfbox{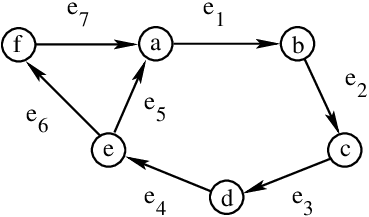}}
\caption{Digraph $G$ of Example \ref{EX3-new}.}
\label{F-ex-d}
\end{figure}

\begin{lemma}
\label{L01}
There is a directed path-decomposition for $G$ of width $p-1$ if and only if 
there is a processing of $Q_G$ with at most $p$ stack-up places.
\end{lemma}

\begin{proof}
By Proposition \ref{prop} we know that $G = G_{Q_G}$.
If there is a directed path-decomposition for $G= G_{Q_G}$ of width $p-1$ 
then by Theorem \ref{th-pd-q} there is a processing of $Q_G$ with at most 
$p$ stack-up places. If there is a processing of $Q_G$ with at most 
$p$ stack-up places then by Theorem \ref{th-q-pd} there is directed 
path-decomposition for $G$ of width $p-1$. 
\end{proof}

\begin{theorem}
Given a list $Q = (q_1, \ldots, q_k)$ of $k$ sequences of bins
and some positive integer $p$. The problem to decide whether there is
a processing of $Q$ with at most $p$ stack-up places is NP-complete.
\end{theorem}

\begin{proof}
The given problem is obviously in NP. Determining
whether the pathwidth of some given (undirected) graph  is 
at most some given value $w$ is NP-complete  \cite{KF79} 
and for symmetric digraphs a special case of the problem on 
directed graphs (cf. Introduction of \cite{KKKTT12}). Thus
the NP-hardness follows from Lemma \ref{L01}, 
because $Q_G$ can be constructed from $G$ in linear time. 
\end{proof}

By the results of  \cite{Bod93} 
it is shown in \cite{MS88} that determining
whether the pathwidth of some given (undirected) graph $G$ is at most some given value $w$
remains NP-complete even for planar graphs with maximum vertex
degree 3. Thus the problem to decide, whether the directed pathwidth 
of some given symmetric digraph $G$ is at most some given value $w$
remains NP-complete even for planar digraphs with maximum vertex
in-degree 3 and maximum vertex out-degree 3.
Thus, by our transformation of graph $G$ into $Q_G$ we get sequences
that contain together at most 6 bins per pallet. Hence, the FIFO stack-up
problem is NP-complete even if the number of bins per pallet is bounded.
Therefore we have proved the following statement.

\begin{corollary}
Given a list $Q = (q_1, \ldots, q_k)$ of $k$ sequences of bins and some
positive integer $p$. The problem to decide whether there is a processing
of $Q$ with at most $p$ stack-up places is NP-complete, even if the
sequences of $Q$ contain together at most 6 bins per pallet.
\end{corollary}

\section{Bounded FIFO stack-up systems}

In this section we show that the FIFO stack-up problem
can be solved in polynomial time, if the number $k$ of 
sequences or the number $p$ of stack-up places is assumed 
to be fixed.

\subsection{Fixed number of sequences}

In this section we assume that the number $k$ of sequences is fixed.

In Section \ref{SC2} we have shown that the FIFO stack-up problem
can be solved by dynamic programming using the processing graph
in time $\bigo(k\cdot (N+1)^{k})$. Thus, we have already shown the following
result.

\begin{theorem}
Given a list $Q = (q_1, \ldots, q_k)$ of sequences of bins  for some fixed $k$ and
a number of stack-up places $p$. The question whether the sequences can
be processed with at most $p$ stack-up places  can be solved in polynomial 
time.
\end{theorem}

Next we improve this result.

\begin{theorem} \label{TH2}
Given a list $Q = (q_1, \ldots, q_k)$ of sequences of bins for some fixed $k$ and
a number of stack-up places $p$. The question whether the sequences can
be processed with at most $p$ stack-up places is non-deterministically
decidable using logarithmic work-space.
\end{theorem}

\begin{proof}
We need $k+1$ variables, namely $pos_1, \ldots, pos_k$ and $open$. Each
variable $pos_i$ is used to store the position of the bin which has been
removed last from sequence $q_i$. Variable $open$ is used to store the number
of open pallets. These variables take $(k+1) \cdot \lceil \log(n) \rceil$
bits.
The simulation starts with $pos_1 := 0, \ldots, pos_k := 0$ and $open := 0$.
\begin{enumerate}[(i)]
\item Choose non-deterministically 
any index $i$ and increment variable $pos_i$. Let $b$ be the
  bin on position $pos_i$ in sequence $q_i$, and let $t := \PL(b)$ be the
  pallet symbol of bin $b$. \\
  {\it Comment:} The next bin $b$ from some sequence $q_i$ will be removed.
\item If $\FI(q_j, t) > pos_j$ or $t\not\in \PLS(q_j)$ for each $j \neq i$, $1 \leq j \leq k$,
  and $\FI(q_i, t) = pos_i$, then increment variable $open$. \\
  {\it Comment:} If the removed bin $b$ was the first bin of pallet $t$ that
  ever has been removed from any sequence, then pallet $t$ has just been
  opened.
\item If $\LA(q_j, t) \leq pos_j$ or $t\not\in \PLS(q_j)$ for each $j$, $1 \leq j \leq k$, then
  decrement variable $open$. \\
  {\it Comment:} If bin $b$ was the last one of pallet $t$, then pallet $t$
  has just been closed.
\end{enumerate}
If $open$ is set to a value greater than $p$ in Step (ii) of the
algorithm then the execution is stopped in a non-accepting state. To
execute Steps (ii) and (iii) we need a fixed number of additional variables.
Thus, all steps can be executed non-deterministically using logarithmic
work-space. 
\end{proof}

Theorem \ref{TH2} implies that the FIFO stack-up problem with a fixed
number of given sequences can be solved in polynomial time since NL is
a subset of P. The class NL is the set of problems decidable
non-deterministically on logarithmic work-space. For example, reachability
of nodes in a given graph is NL-complete, see \cite{Pap94}.
Even more, it can be solved in parallel in polylogarithmic time with
polynomial amount of total work, since NL is a subset of NC$_2$. The class
NC$_2$ is the set of problems decidable in time $\bigo(\log^2(n))$ on a
parallel computer with a polynomial number of processors, see \cite{Pap94}.

\subsection{Fixed number of stack-up places}

In \cite{Tam11} it is shown that the problem of determining the bounded 
directed pathwidth of a digraph is solvable in polynomial time. 
By  Theorem \ref{th-q-pd} and  Theorem \ref{th-pd-q} 
the FIFO stack-up problem with fixed number $p$ of stack-up places is also 
solvable in polynomial time.

\begin{theorem}
Given a list $Q = (q_1, \ldots, q_k)$ of $k$ sequences of bins and
a fixed number of stack-up places $p$. The question whether the sequences can
be processed with at most $p$ stack-up places  can be solved in polynomial 
time.
\end{theorem}


\section{Conclusions and Outlook}

In this paper, we have shown that the FIFO stack-up problem is NP-complete
in general, even if in all sequences together there are at most 6 
bins destined for the same pallet. The problem can be solved in polynomial 
time, if the number $k$ of sequences or if the number $p$ of stack-up places 
is assumed to be fixed.

In our future work, we want to find an answers to the following questions.
Our hardness result motivates to analyze the time complexity of the 
FIFO stack-up problem if in all sequences together there are at most 
$c$, $1<c<6$, bins destined for the same pallet.


Further we are interested in online algorithms for 
instances where we only know the first $c$ bins
of every sequence instead of the complete sequences.
Especially, we are interested in the answer to the following
question: Is there a $d$-competitive online algorithm?
Such an algorithm must compute a processing
of some $Q$ with at most $p \cdot d$ stack-up places, if $Q$ can be processed
with at most $p$ stack-up places.

In real life the bins arrive at the stack-up system on the main conveyor
of a pick-to-belt orderpicking system. That means, the distribution of bins
to the sequences has to be computed. Up to now we consider the distribution
as given. We intend to consider how to compute
an optimal distribution of the bins from the main conveyor onto the sequences
such that a minimum number of stack-up places is
necessary to stack-up all bins from the sequences.


\section*{Acknowledgements}
We thank Prof. Dr.~Christian Ewering who shaped our interest into
controlling stack-up systems. Furthermore, we are very grateful to
Bertelsmann Distribution GmbH in G\"utersloh, Germany, for providing
the possibility to get an insight into the real problematic nature
and for providing real data instances.


\bibliographystyle{plain}
\bibliography{/home/gurski/bib.bib}

\begin{thebibliography}{10}

\bibitem{Bod93}
H.L. Bodlaender.
\newblock A tourist guide through treewidth.
\newblock {\em Acta Cybernetica}, 11:1--23, 1993.

\bibitem{Bor98}
A.~Borodin.
\newblock {\em On-line Computation and Competitive Analysis}.
\newblock Cambridge University Press, 1998.

\bibitem{Kos94}
R.~de~Koster.
\newblock Performance approximation of pick-to-belt orderpicking systems.
\newblock {\em European Journal of Operational Research}, 92:558--573, 1994.

\bibitem{FW98}
A.~Fiat and G.J. Woeginger.
\newblock {\em Online Algorithms: The state of the art}, volume 1442 of {\em
  LNCS}.
\newblock Springer-Verlag, 1998.

\bibitem{GRW14a}
F.~Gurski, J.~Rethmann, and E.~Wanke.
\newblock Moving bins from conveyor belts onto pallets using {FIFO} queues.
\newblock In {\em Proceedings of the International Conference on Operations
  Research (OR 2013), Selected Papers}, pages 185--191. Springer-Verlag, 2014.

\bibitem{GRW15a}
F.~Gurski, J.~Rethmann, and E.~Wanke.
\newblock A practical approach for the {FIFO} stack-up problem.
\newblock In {\em Modelling, Computation and Optimization in Information
  Systems and Management Sciences}, volume 360 of {\em Advances in Intelligent
  Systems and Computing}, pages 141--152. Springer, 2015.

\bibitem{JRST01}
T.~Johnson, N.~Robertson, P.D. Seymour, and R.~Thomas.
\newblock Directed tree-width.
\newblock {\em Journal of Combinatorial Theory, Series B}, 82:138--155, 2001.

\bibitem{KF79}
T.~Kashiwabara and T.~Fujisawa.
\newblock {NP}-completeness of the problem of finding a minimum-clique-number
  interval graph containing a given graph as a subgraph.
\newblock In {\em Proceedings of the International Symposium on Circuits and
  Systems}, pages 657--660, 1979.

\bibitem{KKKTT12}
K.~Kitsunai, Y.~Kobayashi, K.~Komuro, H.~Tamaki, and T.~Tano.
\newblock Computing directed pathwidth in ${O}(1.89^n)$ time.
\newblock In {\em Proceedings of International Workshop on Parameterized and
  Exact Computation}, volume 7535 of {\em LNCS}, pages 182--193.
  Springer-Verlag, 2012.

\bibitem{MS88}
B.~Monien and I.H. Sudborough.
\newblock Min cut is {NP}-complete for edge weighted trees.
\newblock {\em Theoretical Computer Science}, 58:209--229, 1988.

\bibitem{Pap94}
C.H. Papadimitriou.
\newblock {\em Computational Complexity}.
\newblock Addison-Wesley Publishing Company, New York, 1994.

\bibitem{RW97}
J.~Rethmann and E.~Wanke.
\newblock Storage controlled pile-up systems, theoretical foundations.
\newblock {\em European Journal of Operational Research}, 103(3):515--530,
  1997.

\bibitem{RW00}
J.~Rethmann and E.~Wanke.
\newblock On approximation algorithms for the stack-up problem.
\newblock {\em Mathematical Methods of Operations Research}, 51:203--233, 2000.

\bibitem{RW01}
J.~Rethmann and E.~Wanke.
\newblock Stack-up algorithms for palletizing at delivery industry.
\newblock {\em European Journal of Operational Research}, 128(1):74--97, 2001.

\bibitem{Tam11}
H.~Tamaki.
\newblock A {P}olynomial {T}ime {A}lgorithm for {B}ounded {D}irected
  {P}athwidth.
\newblock In {\em Proceedings of Graph-Theoretical Concepts in Computer
  Science}, volume 6986 of {\em LNCS}, pages 331--342. Springer-Verlag, 2011.

\bibitem{Yam09}
K.L. Yam, editor.
\newblock {\em The Wiley Encyclopedia of Packaging Technology}.
\newblock John Wiley \& Sons, New York, 2009.

\end{thebibliography}

\end{document}